\documentclass[12pt]{article}

\usepackage[margin=1in]{geometry}

\usepackage{amsmath, amsthm, amsfonts}
\usepackage{mathtools}

\usepackage{url}
\usepackage{todonotes}
\setuptodonotes{inline}

\usepackage{enumerate}

\usepackage{tikz-cd}
\usepackage{amssymb}

\newtheorem{theorem}{Theorem}
\newtheorem{corollary}[theorem]{Corollary}
\newtheorem{lemma}[theorem]{Lemma}

\newcommand{\FO}{\mathsf{FO}}

\newcommand{\wh}[1]{\widehat{#1}}

\renewcommand{\phi}{\varphi}

\usepackage{xcolor}

\newtheorem{definition}[theorem]{Definition}

\newcommand{\dist}{\mathrm{dist}}

\theoremstyle{definition}

\newcommand{\PP}{\mathcal{P}}

\newcommand{\C}{\mathcal{C}}
\newcommand{\D}{\mathcal{D}}
\newcommand{\E}{\mathcal{E}}
\newcommand{\F}{\mathcal{F}}

\newcommand{\Q}{\mathcal{Q}}

\newcommand{\GG}{\mathcal{G}}

\newcommand{\T}{\mathsf{T}}
\newcommand{\RT}{\mathsf{R}}
\newcommand{\ST}{\mathsf{S}}

\newcommand{\N}{\mathbb{N}}

\newcommand\blfootnote[1]{%
  \begingroup
  \renewcommand\thefootnote{}\footnote{#1}%
  \addtocounter{footnote}{-1}%
  \endgroup
}

\newcommand{\ERCagreement}{
\blfootnote{\noindent
{\begin{minipage}[t]{0.70\textwidth}
\vspace{-15pt}
\small MP is supported by the project {\sc{BOBR}} that have received funding from the European Research Council (ERC) under the European Union's Horizon 2020 research and innovation programme (grant agreement No948057). 

JG is supported by the Polish National Science Centre SONATA-18 grant number 2022/47/D/ST6/03421.
\vspace{10pt}
 \end{minipage} \hspace{5pt}
 \begin{minipage}{.2\textwidth} \vspace{5pt} \hspace{15pt}
 \includegraphics[width=0.8\textwidth]{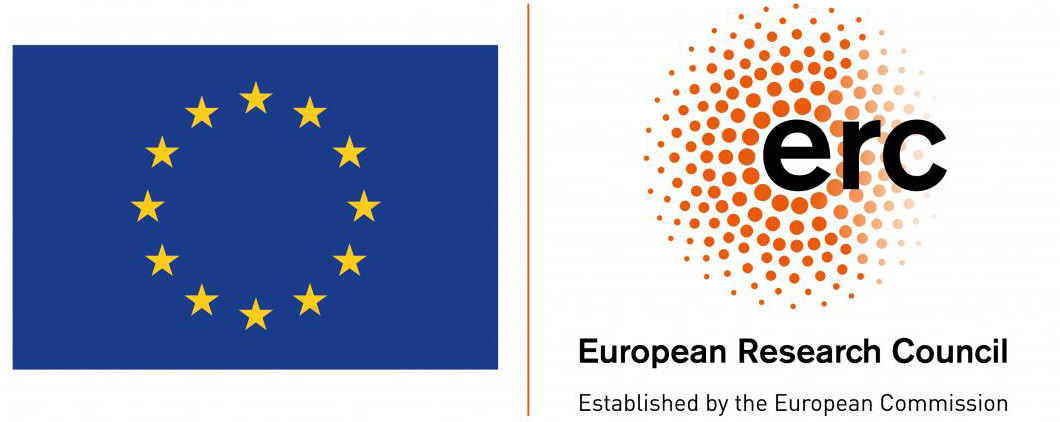}\end{minipage}\hfill}}}

\title{3D-grids are not transducible from planar graphs}

\author{
Jakub Gajarský\thanks{University of Warsaw, Poland. Email: \tt{gajarsky@mimuw.edu.pl}.}\and
Michał Pilipczuk\thanks{University of Warsaw, Poland. Email: \tt{michal.pilipczuk@mimuw.edu.pl}.}\and
Filip Pokrývka\thanks{No affiliation. Email: \tt{fpfilko@gmail.com}.}}

\date{}

\begin{document}
\maketitle

\ERCagreement

\begin{abstract}
    We prove that the class of 3D-grids is cannot be transduced from planar graphs, and more generally, from any class of graphs of bounded Euler genus. 
    To prove our result, we introduce a new structural tool called slice decompositions, and show that every graph class transducible from a class of graphs of bounded genus is a perturbation of a graph class that admits slice decompositions. 
\end{abstract}


\section{Introduction}

A recent topic in  structural and algorithmic graph theory is the study of graph classes through the lens of \emph{first-order transductions}. These are logic-based graph transformations that allow us to create new graphs from old ones, and by extension, new graph classes from old ones. 
Roughly speaking, a transduction can be described as follows. We start with a graph $G$ and color its vertices with a bounded number of color in an arbitrary fashion. Then we use a first-order formula $\psi(x,y)$ to define a new graph $\psi(G)$ on the same vertex set as $G$ and with edge set $\{ uv~:~G\models \psi(u,v)\}$; here we note that the formula $\psi$ can speak about the colors introduced in the first step. Finally, we take an induced subgraph of $\psi(G)$. A transduction $\T$ is therefore determined by a number of colors used and formula $\psi$.Since the coloring step and taking induced subgraph are non-deterministic operations, a transduction actually produces a set of graphs from a source graph $G$ instead of a single graph. We  therefore write $H \in \T(G)$ to denote that $H$ is a result of transduction $\T$ applied to $G$ (instead of $H = \T(G)$). The notion of transduction is naturally extended to graph classes -- the graph class $\T(\C)$ is obtained by applying $\T$ to every graph in $\C$ and taking the union of the results. Finally, we say that a graph class $\D$ is \emph{transducible} from a class $\C$ if there exists a transduction $\T$ such that $\D \subseteq \T(\C)$; we denote this situation by $\D \preceq_{\FO} \C$. This notation is justified by the fact that if $\D$ is transducible from $\C$, then every graph from $\D$ can be encoded in some graph from $\C$ by a first-order formula, and therefore we can think of $\D$ being logically `simpler' than $\C$. We note that in some cases we have $\D \preceq_{\FO} \C$ and $\C \preceq_{\FO} \D$; in this case $\C$ and $\D$ are thought to be equally rich from the perspective of FO logic.

In this work, we are interested in the following question: For which classes $\C$ and $\D$ of graphs we have that $\D$ is transducible from $\C$?
The question is closely related to the
\emph{first-order model checking problem}. In this problem we are given as input a finite graph $G$ and a first-order sentence $\phi$, and the task is to determine whether $G \models \phi$. The relevance of transductions for this problem stems from two different ways how we can specify (families of) graph classes via transductions:
\begin{enumerate}[(i)]
    \item For a graph class $\C$ and a transduction $\T$, we can consider the graph class $\T(\C)$. Here one should imagine that $\C$ is a graph class that we understand well, and we want to understand $\T(\C)$.  The underlying idea is that if $\C$ has nice (structural, logical, algorithmic) properties, then it is reasonable to expect some of these properties will carry over to $\T(\C)$. Note that this also gives us information about all graph classes $\D$ for which we have $\D \subseteq \T(\C)$, that is, all classes with $\D \preceq_{\FO} \C$.
    In fact, one usually considers not just a single graph class $\C$ and single transduction, but whole families of graph classes and all transductions. An important example of this
    situation is the case when we start with the family of \emph{nowhere dense} graph classes, and consider all graph classes that can be obtained from such classes by an arbitrary transduction. Such classes are known as \emph{structurally nowhere dense} graph classes. A recent result of Dreier, M\"ahlmann and Siebertz~\cite{str_nwd} established that the first-order model checking can be solved efficiently on such graph classes. 
    \item For a class $\D$ of graphs, we consider the family all classes $\C$ such that for \emph{no} transduction $\T$ we have that $\D \subseteq \T(\C)$, that is, we study graph classes $\C$ such that $\C \not\preceq_{\FO} \D$.
    Prominent examples of specifying families of graph classes via this approach are \emph{monadically stable} graph classes (here $\D$ is the class of all \emph{half-graphs}) and \emph{monadically dependent} graph classes (here $D$ is the class of all graphs).
    Another way of looking at specifying graph classes in this way is that for every transduction $\T$ there is a graph $G \in \D$ such that $G \not\in \T(\C)$. Studying graph classes determined in this way therefore boils down to studying what graphs \emph{cannot} be transduced from $\C$. The rationale behind this approach is that understanding these kind of obstructions for $\C$ can lead to establishing structural results about $\C$, and consequently to algorithmic results. These kind of ideas played a role in recent important results~\cite{tww_ordered,MC_stable,dependent_dichotomy}.
\end{enumerate}


Currently, in the realm of the first-order model checking problem, the focus is on monadically dependent graph classes, that is, graph classes $\C$ such that $\GG \not\preceq_{\FO} \C$, where $\GG$ is the class of all graphs. It has been conjectured~\cite{NIPconjecture,bd_deg_interp2} that this is all that is needed to obtain an efficient model checking algorithm. However, the general question when $\D \preceq_{\FO} \C$ for various graph classes $\D$ simpler than $\GG$ is well-motivated and can lead to interesting progress in structural graph theory, as we now argue.

If we want to prove that a graph class $\D$ is transducible from a graph class $\C$, it is enough to exhibit a transduction $\T$ such that $\D \subseteq \T(\C)$. While this is conceptually simple, it requires us to understand both $\C$ and $\D$ -- we need expose enough structure and simplicity in graphs from $\D$ and at the same time enough richness and complexity in graphs from $\C$, so that we can encode each $G \in \D$ in a graph from $\C$.

On the other hand, to prove that $\D$ is not transducible from $\C$, one typically proceeds as follows. Let $\Pi$ be a property a graph class can have (for example, having bounded treewidth or having bounded clique-width). We will think of $\Pi$ as the family of all graph classes that have property $\Pi$ and write $\C \in \Pi$ instead of ``$\C$ has property $\Pi$". We say that a class $\Pi$ is a \emph{transduction ideal} if $\C \in \Pi$ implies that $\T(\C) \in \Pi$ for all transductions. We also sometimes say that $\Pi$ is a \emph{transduction invariant property}. Important examples of such properties are having bounded shrub-depth~\cite{shrubdepth}, rank~\cite{rank}, clique-width~\cite{cw}, twin-width~\cite{tww1journal}, flip-width~\cite{flipwidth}, or being monadically stable~\cite{shelah1971stability} or monadically dependent~\cite{mdependence}. Non-examples are 
 having bounded treewidth or being planar. Then, to show for example that the class of all 2-dimensional grids is not transducible from trees, where we can take $\Pi$ to be the transduction ideal ``having bounded clique-width". 
However, it is not possible to use currently known transduction ideals to answer many (even basic) non-transducibility questions (see below). Asking such questions will presumably lead to identifying new transduction ideals.
Given how important the currently known transduction ideals are in structural graph theory and model theory, this could lead to discoveries of new important graph theoretic notions and enrich our understanding of the connection between structural graph theory and (finite) model theory.




While, as explained above, studying graph classes and relationships between them using transductions is well-motivated (see also the new survey~\cite{survey_lens}),
 our current understanding of the relation $\preceq_{\FO}$ is fairly limited. One easily checks that the relation $\preceq_{\FO}$ is a quasi-order (known as the \emph{first-order transduction quasi-order}), and its order-theoretic properties together with some other fine-grained questions were studied in~\cite{trans_quasi-order}.
However, as mentioned above, many basic questions are currently unanswered. For example:
\begin{itemize}
    \item Is it true that for every $k$ the class of all graphs of graphs of treewidth $k+1$ is transducible from the class of graphs of treewidth at most $k$? The answer is known  for pathwidth \cite{trans_quasi-order}.
    \item Is it true that the class of all toroidal graphs is transducible from the class of planar graphs?
    \item Is it true that for every $k$ the class of graphs of treewidth $k$ is  transducible from planar graphs? 
\end{itemize}


Presently, probably the biggest bottleneck to obtaining a fine understanding of the first-order transduction quasi-order is the lack of a robust toolbox for proving non-tranducibility results that extends beyond commonly studied transduction invariant properties. Overcoming this lack of results and techniques is the motivation behind our research.


\paragraph{Our results.} Our main result is the following. In its statement, a \emph{cube} is a 3-dimensional grid with all sides of equal length.

\begin{theorem}
\label{thm:main}
    The class of all cubes is not transducible from any class of graphs of bounded genus.
\end{theorem}

To prove Theorem~\ref{thm:main}, we introduce a new tool called a  \emph{slice decomposition} of a graph. Fix a function $f:\N \to \N$ and let $G$ be a graph. A slice decomposition of $G$ guarded by $f$ is a partition $V_1 \cup \ldots V_\ell$ of $V(G)$ with the following two properties:
\begin{itemize}
    \item Every edge of $G$ is either within some part $V_i$ or or between consecutive parts $V_i$, $V_{i+1}$.
    \item Any $k$ consecutive parts induce a subgraph of $G$ of clique-width at most $f(k)$.
\end{itemize}
We say that a class $\C$ of graphs \emph{admits slice decompositions} if there exists a function $f$ such that every $G\in\C$ has a slice decomposition guarded by $f$.

Our structural result for graph classes interpretable in classes of bounded genus is Theorem~\ref{thm:slices} below. In its statement, a graph $H$ is a $k$-flip of graph $G$ if $H$ can be obtained from $G$ by partitioning $V(G)$ into at most $k$ classes and complementing (flipping) the adjacency between some pairs of classes.
\begin{theorem}
\label{thm:slices}
    Let $\C$ be a class of graphs of bounded genus and let $\T$ be a transduction. Then there exists a class $\D$ of graphs that has slice decompositions and number $k \in N$ such that every graph in $\T(\C)$ is a $k$-flip of a graph from $\D$.
\end{theorem}
The proof of Theorem~\ref{thm:slices} relies on the product structure theorem for graphs of bounded genus and Gaifman's locality theorem. We note that our proof of Theorem~\ref{thm:main} does not use the full power of Theorem~\ref{thm:slices}, as we only use the fact that the clique-width of each part $V_i$ in the slice decomposition is bounded. Theorem~\ref{thm:slices} can therefore be useful for proving stronger non-transducibility results.


\subsection{Outline of the proof}

Our proof of Theorem~\ref{thm:main} can be broken down into the following steps:
\begin{enumerate}
    \item We will show that if $\C$ is such that $\Q \subseteq \T(\C)$ for some transduction $\T$, then there is a transduction $\T'$ of bounded range such that $\Q \subseteq \T'(\C)$ (Section~\ref{sec:reduction}).
    \item We will show that if $\T$ is a transduction of bounded range and $\C$ is a class of graphs of bounded Euler genus, then $\T(\C)$ admits slice decompositions for some function $f$ (Section~\ref{sec:decompositions}). 
    \item Assuming that cubes can be transduced from a class of graphs of bounded Euler genus, by taking $\C$ as the class of planar graphs in the previous two statements we get that the class of cubes admits slice decompositions for some function $f$.
    \item From a result of Berger, Dvor{\'{a}}k and Norine~\cite{Dvorak_grids} we deduce that the class of all 3D-grids does not admit slice decompositions, which gives us a contradiction (Section~\ref{sec:finish}).
\end{enumerate}

\section{Preliminaries}

For two integers $m,n$ with $m < n$ we denote by $[m,n]$ the set $\{m, m+1, \ldots, n\}$. We denote by $[n]$ the set $[1,n]$.

\subsection{Logic}
\label{sec:logic}

We assume familiarity with basic notions of the first-order logic: signatures, formulas, quantifier rank.

Graphs are modeled as structures over a single binary predicate symbol $E$. 
Often, we will speak of \emph{colored} graphs -- these are graphs equipped with finitely many colors modeled as unary predicates. To each class of colored graphs, we associate a signature $\Sigma$ which consists of the binary predicate symbol $E$ and finitely many unary predicate symbols. Unless explicitly specified otherwise, all classes of structures considered in this paper will be (colored) graphs, and we will therefore use letters $G,H, \ldots$ for structures and $V(G), V(H), \ldots$ for their underlying universes.

\paragraph{Interpretations and transductions.}
We say that a formula $\psi(x,y)$ is \emph{antireflexive} and \emph{symmetric} if for all graphs $G$ and $u,v \in V(G)$ we have $G \not\models \psi(u,u)$ and $G \models \psi(u,v) \Leftrightarrow G \models \psi(v,u)$. 
For a antireflexive and symmetric formula $\psi(x,y)$ and a graph $G$ we denote by $\psi(G)$ the graph on vertex set $V(G)$ and edge set 
$$ E(\psi(G)) = \{ uv~|~G \models \psi(u,v) \} $$
We then say that $\psi(x,y)$ is an \emph{interpretation formula} and that the graph $\psi(G)$ is \emph{interpreted in $G$}.

A \emph{transduction} is an operation determined by a number $m$ and an interpretation formula $\psi(x,y)$ that from a graph $G$ produces a new graph as follows:
\begin{enumerate}[(i)]
	\item First, vertices of $G$ are marked with (at most) $m$ unary predicates in arbitrary fashion to obtain graph $G^{+}$.
	\item Then, we apply $\psi$ to $G^{+}$ to produce graph $\psi(G)$ (the formula $\psi$ can use the unary predicates introduced in the first step).
	\item Finally, we take an induced subgraph of $\psi(G^{+})$.
\end{enumerate} 
Since the steps $(i)$ and $(iii)$ are not deterministic, transduction actually does not produce a single graph but a set of graphs. If $\T$ is a transduction, we denote by $\T(G)$ the set of graphs it produces, and if $H \in \T(G)$, we say that $H$ is \emph{transduced} from $G$. For a class of graphs $\C$ and a transduction $T$, we set $\T(\C):=\bigcup_{G \in C} \T(G)$.

We say that a formula  $\psi(x,y)$ has \emph{range} $d$ if for every graph $G$ and every $u,v \in V(G)$ we have that $G \not\models \psi(u,v)$ implies $dist_G(u,v) \le d$. This means that when creating the graph $\psi(G)$ from $G$, the formula $\psi(x,y)$ will not create an edge between two vertices that are far away in $G$. We say that a formula  $\psi(x,y)$ has \emph{bounded range} if there exists a bound $d$ such that $\psi(x,y)$ has range $d$. Finally, we say that a transduction $T$ has bounded range, if the interpretation formula used by $\T$ has bounded range.

It is easily verified that a composition of two transductions is a transduction; we denote the composition of transductions $S$ and $\T$ by $\T \circ \ST$. Moreover, if both $\ST$ and $\T$ are of bounded range, then so is $\T \circ S$ (namely, if $\ST$ has range $d_1$ and $\T$ has range $d_2$, then $\T \circ \ST$ has range $d_1d_2$).


\paragraph{Gaifman's theorem}
Several of our results rely on Gaifman's locality theorem.
We say that a first-order formula $\gamma(x_1,\ldots, x_k)$ is $r$-\emph{local} if for every $G$ and every $v_1,\ldots,v_k \in V(G)$ we have that $G \models \gamma(v_1,\ldots,v_k) \Leftrightarrow B^G_r(v_1,\ldots,v_k) \models \gamma(v_1,\ldots,v_k)$ where $ B^G_r(v_1,\ldots,v_k) = \bigcup_{1 \le i \le k} B^G_r(v_i) $.
A \emph{basic local sentence} is a sentence of the form
$$ \tau \>\equiv\>
	\exists x_1 \ldots \exists x_k	\left(\bigwedge_{1 \le i < j \le k} 
		\dist(x_i,x_j) > 2r 
	\land \bigwedge_{1 \le i \le k} \alpha(x_i) \right)
,$$
where $\alpha(x)$ is $r$-local and $\dist(x,y)>2r$ is a formula expressing that the distance between $x$ and $y$ is larger than $2r$. 


\begin{theorem}[Gaifman's locality theorem \cite{gaifman1982local}]\label{thm:Gaifman}
Every first-order formula $\phi(x_1,\ldots,x_k)$ is equivalent to a boolean combination of a $7^q$-local formula $\rho(x_1,\ldots,x_k)$ and basic
local sentences with locality radius $7^q$, where $q$ is the quantifier rank of $\phi$. 
\end{theorem}

In our proofs we will use the following two corollaries of Gaifman's theorem.
\begin{corollary}
\label{cor:gaifman}
    For every formula $\psi(x,y)$ there exists a number $r$ and $r$-local formula $\rho(x,y)$ such that for every $G$ one of the following is true:
    \begin{itemize}
        \item $\psi(G)$ is a clique
        \item $\psi(G)$ is a edgeless
        \item $\psi(G) = \rho(G)$
        \item $\psi(G) = \lnot\rho(G)$
    \end{itemize}
\end{corollary}

\begin{proof}
    Apply Gaifman's theorem to $\psi(x,y)$ to obtain a boolean combination of $r$-local formula $\rho(x,y)$ and basic local sentences $\tau_1,\ldots, \tau_m$. Let $G$ be arbitrary. After evaluating all $\tau_i$ on $G$, the boolean combination of formulas reduces to $\top$, $\bot$, $\rho(x,y)$ or $\lnot\rho(x,y)$ and the result follows.
\end{proof}


Let $G$ and $H$ be graphs. We say that $H$ is a \emph{$k$-flip} of $G$ if one can obtain $H$ from $G$ by partitioning $V(G)$ into at most $k$ parts $V_1,\ldots,V_k$ and complementing the adjacency between some pairs of parts (not necessarily distinct).

\begin{theorem}[\cite{trans_bd_range}, also implicit in~\cite{gajarsky2022differential}]
\label{thm:gaifman_bd_range}
    Let $\C$ be a class of graphs and $\T$ a transduction. Then there exists a transduction $\RT$ of bounded range such that every graph in $\T(\C)$ is a $k$-flip of a graph from $\RT(\C)$.
\end{theorem}


\subsection{Graph theory}

All graphs in this paper are finite, simple, undirected, and unless specified otherwise, they do not contain loops. 
We use standard graph-theoretic terminology and notation. 

While we will rely on the notions of treewidth and cliquewidth, we will never need to explicitly work with tree-decompositions and cliquewidth decompositions, and so we omit their definitions. We will only need the properties stated in Theorems~\ref{thm:cw_sparse_tw} and~\ref{thm:tw_trans_cw}.





\begin{theorem}[\cite{cw_sparse_tw}]
\label{thm:cw_sparse_tw}
    There exists a function $g: \N \times \N \to \N$ with the following property. If $\C$ is a class of graphs of  clique-width at most $c$ such that no graph from $\C$ contains $K_{s,s}$ as a subgraph, then $\C$ has treewidth at most $g(c,s)$.
\end{theorem}

\begin{theorem}[\cite{cw_interp}]
\label{thm:tw_trans_cw}
    For every formula $\psi(x,y)$ there exists a function $g_{\psi}: \N \to \N$ such that for every graph $G$ we have $cw(\psi(G)) \le g_\psi(tw(G))$.
\end{theorem}

Let $G$ and $H$ be graphs. The \emph{strong product} of $G$ and $H$, denoted by $G \boxtimes H$, is the graph on vertex set $V(G)\times V(H)$ in which there is an edge between distinct vertices $(u,x)$ and $(v,y)$ if
\begin{itemize}
    \item $u=v$ and $xy \in E(H)$
    \item $x=y$ and $uv \in E(G)$
    \item $uv \in E(G)$ and $xy \in E(H)$.
\end{itemize}

In their celebrated result~\cite{queue}, Dujmovic, Joret, Micek, Morin, Ueckerdt and Wood proved the following result, known as the \emph{produc structure theorem for planar graphs}.
\begin{theorem}[\cite{queue}]
    Every planar $G$ graph is a subgraph of $H \boxtimes P$, where $H$ is a graph of treewidth at most $8$ and $P$ is a path.
\end{theorem}
More generally, they have proved the following\footnote{We note that the result of~\cite{queue} actually states that every graph of Euler genus $G$ is subgraph of $(K_{2g} + H) \boxtimes P$, where $H$ is of treewidth $8$. Our version of the theorem follows easily from this, since  the treewidth of $K_{2g} + H$ is at most $2g+8$.}.
\begin{theorem}[\cite{queue}]
\label{thm:product_general}
    Let $\C$ be a class of graphs of bounded Euler genus. Then there exists $t\in \N$ such that every $G \in \C$ is a subgraph of $H \boxtimes P$, where $H$ is a graph of treewidth at most $t$ and $P$ is a path.
\end{theorem}

We will work with a subclass of all 3D-grids in which all sides have the same length. We call such graphs \emph{cubes}.
\begin{definition}
 The \emph{cube with side length $N$} is the graph $Q_N$ with vertex set $[N]^3$ such that for every $i,j,k, i',j',k' \in [N]$ we have $(i,j,k)(i',j',k') \in E(Q_N)$ if and only if $|i-i'| + |j-j'| + |k-k'| = 1$.

 We define the class $\Q = \{Q_N~|~N \in \N\}$.
\end{definition}

We define three projection functions $\pi_1, \pi_2, \pi_r$ from $V(Q_N)$ to $[N]$ in the natural way by setting $\pi_1((i,j,k)) = i$, $\pi_2((i,j,k)) = j$ and $\pi_3((i,j,k)) =k$.

Sometimes we will consider induced subgraphs of cubes that are themselves isomorphic to a cube. For example, consider $Q_{2N}$ for some $N$ and set $S = [N,2N]\times[N,2N]\times[N,2N]$. Then $Q_{2N}[S]$ is isomorphic to the cube $Q_N$, but is not equal to it because it has different vertex set. We will usually ignore this subtlety and treat such induced subgraphs as cubes (thus, we also consider graphs isomorphic to cubes to be cubes).

\section{Slice decompositions}
\label{sec:decompositions}

In this section we define \emph{slice} decompositions and prove that

\begin{definition}
    Let $f: \N \to \N$ be a function and $G$ a graph. A \emph{slice decomposition of $G$ guarded by $f$}, is a partition $V_1, \ldots, V_\ell$ of $V(G)$ such that:
    \begin{enumerate}[(i)]
        \item If $uv \in E(G)$, then $u \in V_i$ and $v \in V_j$ such that $|i-j|\le 1$.
        \item For every $k \in \N$ and every $i \le \ell-k + 1$ we have that clique-width of $G[V_i \cup \ldots \cup V_{i+k-1}]$ is at most $f(k)$.
    \end{enumerate}   
\end{definition}

Let $\D$ be a class of graphs. We say that \emph{$\D$ has slice decompositions} of there exists a function $f$ such that every $G \in \D$ ha a slice decomposition guarded by $f$.

Our main result in this section is the following structural result for graph classes transducible from graph classes of bounded genus.
\begin{theorem}
    Let $\C$ be a class of graphs of bounded genus and let $\T$ be a transduction. Then there exists a class $\D$ of graphs that has slice decompositions and number $k \in N$ such that every graph in $\T(\C)$ is a $k$-flip of a graph from $\D$.
\end{theorem}

By Theorem~\ref{thm:gaifman_bd_range} we know there exists a transduction of bounded range $T'$ and $k \in \N$ such that every graph in $\T(\C)$ is a $k$-flip of some graph from $\T'(\C)$. Thus, it suffices to show the following lemma.
\begin{lemma}
\label{lem:f-decompositions}
    Let $\C$ be a class of graphs of bounded genus and $\T$ be a transduction of bounded range. Then there exists a function $f: \N \to \N$ such that each $G \in \T(\C)$ has a $f$-decomposition.
\end{lemma}

In the proof of Lemma~\ref{lem:f-decompositions} we will use the following observation, which follows easily from the definition of strong product.
\begin{lemma}
\label{lem:tw_product}
For any graph $G$ and number $k$, we have that $tw(G \boxtimes P_k) \le k (tw(G)+1) - 1$.
\end{lemma}

\begin{proof}[Proof of Lemma~\ref{lem:f-decompositions}]
    Let $G \in \C$ be arbitrary and let $d$ be the range of $\T$.
    By Theorem~\ref{thm:product_general} there exists $t$ such that every $G$ is a subgraph of $H \boxtimes P$, where $P$ is a path and $tw(H) \le t$ for some fixed $t$ depending on (the genus of) $\C$. Let $G \in \C$ be arbitrary, and fix $H$ and $P$ such that $G$ is a subgraph of $H \boxtimes P$. Let $\psi(x,y)$ be the interpretation formula of $T$. By Corollary~\ref{cor:gaifman} we know that $\psi(G)$ is a clique, edgeless graph, $\rho(G)$ or $\lnot\rho(G)$, where $\rho$ is a $r$-local formula. If $\psi(G)$ is a clique or edgeless graph, then $cw(\psi(G)) \le 2$ an we can take $V_1 = V(G)$ as our $f$-decomposition of $\psi(G)$. In the rest of the proof we will therefore assume that $\psi(G) = \rho(G)$ or $\psi(G) = \lnot\rho(G)$; both of these cases will be treated in the same fashion, the only important part is that $\rho$ is $r$-local.

    Let $1,\ldots, p$ be the vertices of $P$. For every vertex $v = (a,i)$ of $G$, where $a \in V(H)$ and $i \in V(P)$, we will denote $\pi_1(v) = a$ and $\pi_2(v) = i$. We partition $V(G)$ into sets $U_1,\ldots,U_p$ by setting 
    $U_i = \{(v \in V(G)~|~\pi_2(v) = i\}$. Note that from the definition of strong product and construction of sets $U_i$ it follows that if $\dist_G(u,v) \le d$, then we have that $u \in V_i$, $v \in V_j$ with $|i-j|\le d$.
    Also note that for each $i$ the graph $G[U_i]$ is isomorphic to a subgraph of $H$, and so each $G[U_i]$ has treewidth at most $t$.
    
    We define the $f$-decomposition of $\psi(G)$ (with $f$ to be determined later) as follows: We set 
    $$V_i:= \bigcup_{j=d(i-1)+1}^{di} U_j,$$
    that is, we set $V_1 := U_1 \cup \ldots \cup U_d$, $V_2 := U_{d+1} \cup \ldots U_{2d}$ etc. We now argue that the partition $V_1,\ldots,V_{\lceil\frac{p}{d}\rceil}$ satisfies the conditions from the definition of $f$-decomposition.

    (i) We  argue that if $|i-j| \ge 2$ and $u \in V_i$, $v \in V_j$, then $uv \not\in E(\psi(G))$. From the definition of $V_i$ and $V_j$ it is easy to verify that $u \in U_a$ and $v \in U_b$ for which $|a-b| > d$, and so $dist_G(u,v) > d$. Thus, since $\psi$ is a formula of range $d$, it cannot create and edge between $u$ and $v$.

    (ii) Let $k$ and $i \le \lceil\frac{p}{d}\rceil - k  + 1$ be arbitrary. Set $$S := V_i \cup \ldots \cup V_{i + k - 1} = U_{d(i-1) + 1} \cup \ldots \cup U_{d(i+k-1)} $$
    and $S' := U_{d(i-1) + 1 -r} \cup \ldots \cup U_{d(i+k-1) + r} $. Note that from the definition of $S'$ we have for every $v \in S$ that $N_r^G[v] \subseteq S'$.

    Our goal is to bound the clique-width of $\psi(G)[V_i\cup \ldots \cup V_{i+k -1}] = \psi(G)[S]$ by a function of $k$ (where the function itself can depend on $\C$ and $\psi$). We will proceed by showing that $\psi(G)[S]$ is an induced subgraph of $\psi(G[S'])$ and that $\psi(G[S'])$ has clique-width bounded by a function of $k$. Since clique-width is closed under taking induced subgraphs, this will finish the proof.

    First we will argue that $\psi(G)[S]$ is an induced subgraph of $\psi(G[S'])$. We will assume that $\psi(G) = \rho(G)$; the proof for the case when $\psi(G) = \lnot\psi(G)$ is completely analogous.
    Since $\rho$ is $r$-local, we know that for any $u,v \in V(G)$ we have 
    \begin{equation}
    \label{eq:1}
       G \models \rho(u,v) \Longleftrightarrow G[N_{r'}^G[u]\cup N_{r'}^G[v]] \models \rho(u,v)  
    \end{equation}   
    Similarly we have for $u,v \in S'$ that 
    \begin{equation}
    \label{eq:2}
        G[S'] \models \rho(u,v) \Longleftrightarrow (G[S'])[N_{r'}^{G[S']}[u]\cup N_{r'}^{G[S']}[v]] \models \rho(u,v)
    \end{equation}
    Since for each $u, v\in S$ we have that $N_{r}^G[v] \subseteq S'$ and $N_{r}^G[v] \subseteq S'$, we have that 
    \begin{equation}
       \label{eq:3}
       G[N_{r}^G[u]\cup N_{r}^G[v]] = (G[S'])[N_{r}^{G[S']}[u]\cup N_{r}^{G[S']}[v]] 
    \end{equation}
 
    Thus, for any $u,v \in S$, by combining (\ref{eq:1}) and (\ref{eq:2}) with  (\ref{eq:3}) we get that $uv$ is an edge in $\rho(G)$ if and only if $uv$ is an edge in $\rho(G[S'])$, as desired.
    
    Second, we argue that $\psi(G[S'])$ has bounded clique-width.
    Since $S'$ consists of  $d(i+k-1) + r - (d(i-1) -r) = kd + 2r$ sets $U_i$, we know that $G[S']$ is isomorphic to a subgraph of $H \boxtimes P_{kd + 2r}$, and hence $tw(G[S']) \le (kd + 2r)(tw(H)+1) - 1$ by Lemma~\ref{lem:tw_product}. By Theorem~\ref{thm:tw_trans_cw} we know that $cw(\psi(G[S'])) \le g(tw(G[S']))$ for some function depending on $\psi$.
    Thus we can set $f(k) := g((kd + 2r)(t+1)-1)$. Since $t$ depends on the graph class $\C$, and $d$ and $r$ depend on $\psi(x,y)$, this finishes the proof.
\end{proof}

\section{Reduction to transductions of bounded range}
\label{sec:reduction}

Our goal in this section will be to prove the following theorem, which is a crucial ingredient in the proof of our main result, Theorem~\ref{thm:main}.

\begin{theorem}
\label{thm:reduction_bd_range}
   Let $\C$ be a class of graphs such that $\Q \subseteq \T(\C)$ for some transduction $\T$. Then there exists a transduction $\T'$ of bounded range such that $\Q \subseteq \T'(\C)$.
\end{theorem}

We now briefly outline the proof of Theorem~\ref{thm:reduction_bd_range}. Let $\C$ be a class of graphs and $\T$ a transduction such that $\Q \subseteq \T(\C)$. By Theorem~\ref{thm:gaifman_bd_range} we know that there exist $k \in \N$ and a transduction $\RT$ of bounded range such that every $Q_N$ is a $k$-flip of some graph in $\RT(\C)$. We will show that there is a bounded range transduction $\ST$ such that $\Q \subseteq \ST(\RT(\C))$. Then, by setting $\T':=\ST\circ \RT$ and noting that a combination of bounded-range transductions is a bounded-range transduction, we obtain Theorem~\ref{thm:reduction_bd_range}. 

To find transduction $\ST$, we analyze how  we can obtain a large cube $Q_N$ from a graph $G$ by complementing adjacency between some of $k$-parts of a partition of $V(G)$. To perform this analysis, we first extend $G$ with the precise information about how we obtain $Q_N$ from it; this leads to the notion of \emph{$k$-flip structure} (see the next section). We then analyse $k$-flip structures that produce cubes, and establish some of their properties (Section~\ref{sec:flip_str_properties}). Using these properties, we show (Lemma~\ref{lem:flip_bd_range}) that the graphs $\RT(\C)$ either contain large cubes as induced subgraphs (in this case transduction $\ST$ does not have to do anything, just take na induced subgraph) or they are of bounded diameter (in which case the transduction that prodices the flip of $G$ is already of bounded range).

\subsection{Flip-structures}

Let $H$ be a graph that can have loops on its vertices. We say that a vertex $v$ of $H$ is \emph{isolated} if it does not have a neighbor in $H$ and does not have a loop. We denote the set of isolated vertices of $H$ by $I(H)$.

\begin{definition}
A \emph{$k$-flip structure} is a triple $(G, \PP, H)$ where:
\begin{itemize}
    \item $G$ is a graph.
    \item $\PP = \{V_1,\ldots, V_\ell\}$ is a partition of $V(G)$ where we allow some parts to be empty.
    \item $H$ is a graph with with vertex set $\{1, \ldots, k\}$ that can have loops, and in which we have that if $V_i = \emptyset$, then  $i \in I(H)$.
\end{itemize}
\end{definition}

We remark that the fact that we allow for some parts $V_i$ to be empty is for technical convenience; this is useful because when we pass to $k$-flip substructure (see below), some parts may become empty. Also, we force the indices of empty parts to be isolated in $H$ to guarantee that all vertices of every nontrivial connected components of $H$ correspond to non-empty parts.

We define function $\rho: V(G) \to [k]$ by setting $\rho(v) = i$, where $i$ is the unique part $V_i \in \PP$ such that $v \in V_i$.
To every $k$-flip structure $(G, \PP, H)$ we associate the graph $F(G, \PP, H)$ with vertex set $V(G)$ as follows: 
We go through all pairs of distinct vertices $u,v$ of $G$ and complement (flip) the adjacency between them if $u \in V_i$ and $v \in V_j$ and $ij$ is an edge of $H$. In symbols, we set the edge set of $F(G, \PP, H)$ to be the set
$$\{uv~|~ u\not=v \land (uv \in E(G) \oplus \rho(v)\rho(u) \in E(H))\}$$
where $\oplus$ denotes the XOR operation. 
For a class $\C$ of flip structures we define 
$F(\E) = \{F(G, \PP, H)~|~(G, \PP, H) \in \E\}$.

Let $(G,\PP,H)$ be a $k$-flip structure.
A \emph{$k$-flip substructure} of $(G,\PP,H)$ is a $k$-flip structure $(G',\PP',H')$ where $G'$ is an induced subgraph of  $G$, each $V_i' \in \PP'$ is defined by $V_i':=V_i \cap V(G')$ and in which $H'$ is obtained from $H$ by isolating every $i \in [k]$ with $V_i' = \emptyset$. In particular, every induced subgraph $G'$ of $G$ determines a $k$-flip substructure $(G',\PP',H')$. It is easily checked that if $A \subseteq V(G)$ and  $G' = G[A]$, then $F(G,\PP,H)[A] = F(G',\PP',H')$.

Finally, we say that a class $\E$ of $k$-flip structures is \emph{hereditary} if for every $(G,\PP,H) \in \E$ and every induced $k$-flip substructure $(G',\PP',H')$ of $(G,\PP,H)$ we have $(G',\PP',H') \in \E$.

Before we proceed further, we make some remarks about how we will think of $k$-flip structures in the rest of this section.
\begin{itemize}
    \item Let $(G,\PP,H)$ be a $k$-flip structure and consider the subset $S$ of $V(G)$ given by $S:=\cup_{i \in I(H)} V_i$. Since vertices from $I(H)$ are isolated in $H$, the sets $V_i$ with $i \in I(H)$ do not take part in any flips, and so $S$ is the set of vertices that are not affected when going from $G$ to $F(G,\PP,H)$. This means that we have $G[S] = F(G,\PP,H)[S]$.
    \item The set $\Bar{S}$ given by $V(G) \setminus S$ is the set of vertices whose adjacency is affected when going from $G$ to $F(G,\PP,H)$. This set can be further partitioned as follows. For any connected component $C$ of $H - I(H)$, let $U(C):=\cup_{i \in V(C)} V_i$. Let $C_1,\ldots, C_m$ be the connected components of $H - I(H)$. Then we have $\Bar{S} = \bigcup_{j=1}^m U(C_j)$. We will later see (Lemma~\ref{lem:properties}) that the sets $U(C_j)$ induce subgraphs of $G$ of small diameter.
    \item Finally, remark about notation: In the rest of the section we sometimes treat connected components of $H - I(H)$ as sets of vertices instead of graphs to unclutter the notation, so we will for example we write $\cup_{i \in C} V_i$ instead of  $\cup_{i \in V(C)} V_i$.
\end{itemize}




\subsection{Properties of flip structures with $F(G,\PP,H) = Q_N$}
\label{sec:flip_str_properties}

In this section we establish several useful properties of $k$-flip structures that produce cubes $Q_N$. We start by a simple observation that identifies a situation when a large cube is already present in $G$ and remains there after performing the flips.  
\begin{lemma}
\label{lem:induced_subgrid}
Let $(G,\PP,H)$ be such that $F(G,\PP,H) = Q_N$ and let $r \le N$. If there is a vertex $v \in V(G)$ such that $N_{r}^G[v] \subseteq \cup_{i \in I(H)}V_i$, then $G$ contains $Q_{\lfloor r/3\rfloor }$ as an induced subgraph.    
\end{lemma}
\begin{proof}
  The lemma follows easily from the following two observations.

Observation 1: Let $N$ be arbitrary and let $r$ be such that $3r \le N$. Let $v$ be an arbitrary vertex of $V(Q_N)$. Then $N^{Q_N}_{r}[v]$ contains $Q_{\lfloor r/3\rfloor }$ as an induced subgraph.

Observation 2: Let $(G,\PP,H)$ be a $k$-flip structure and set $S = \cup_{i \in I(H)} V_i$. Then $G[S] = F(G,\PP, H)[S]$, since the vertices from $S$ are not affected by the operation $F$ that takes $G$ to $F(G,\PP,H)$.  
\end{proof}

The next three lemmas will be used in the proof of Lemma~\ref{lem:properties} that will let us focus only on flip-structures with desirable properties. 

\begin{lemma}
\label{lem:only_large_classes}
Let $(G,\PP,H)$ be a $k$-flip structure such that $F(G,\PP,H) = Q_{3N}$. Assume that $V_i$ is a part of $\PP$ with $|V_i| \le 12$. Then there exists a $k$-flip substructure $(G',\PP',H')$ of $(G,\PP,H)$ with $V_i' = \emptyset$ such that $F(G',\PP',H') = Q_N$.
\end{lemma}
\begin{proof}
 We partition $Q_{3N}$ into $27$ cubes of side length $N$ as follows: Set $I_1 = [1,N]$, $I_2 = [N+1, 2N]$, $I_3 = [2N+1, 3N]$. For each $a,b,c \in \{1,2,3\}$ consider the set $S^{abc}$ that consists of vertices $v \in V(Q_{3N})$ with $\pi_1(v) \in I_a$, $\pi_2(v) \in I_b$ and $\pi_3(v) \in I_c$. Then each $Q^{abc} := Q_{3N}[S^{abc}]$ is a cube of side length $N$. 
 
Since $|V_i| \le 12$, by pigeonhole principle there are $a,b,c \in  \{1,2,3\}$ such that $S^{abc} \cap V_i = \emptyset$. Then $G' = G[S^{abc}]$ determines a $k$-flip substructure $(G', \PP', H')$ of $(G,\PP,H)$ such that $F(G', \PP', H') = Q^{abc}$.
\end{proof}

\begin{lemma}
\label{lem:small_dist}
Let $(G,\PP,H)$ be such that $F(G,\PP,H) = Q_N$. Let $i,j$ be such that $ij \in E(H)$, $|V_i| > 12$ and $|V_j| > 12$. Then we have $dist_G(u,v) \le 3$ for every $u,v \in V_i\cup V_j$.
\end{lemma}
We remark that in the lemma we do not require that $i \not=j$; this will not play a role in the proof.
\begin{proof}
We will prove the lemma for two of the four possible cases: (i) $u,v \in V_i$ and (ii) $u \in V_i, v \in V_j$. The arguments for the other two cases $u,v \in V_j$ and $u \in V_j, v \in V_i$ are symmetric.

(i) Assume that $u,v \in V_i$. 
We know that in $Q_N = F(G,\PP,H)$ both vertices $u$ ad $v$ have degree at most $6$, and so in particular they both have at most $6$ neighbors in $Q_N[V_j]$. Since $|V_j| > 12$, there exists $w \in V_j$ such that both $u$ and $v$ are not adjacent to $w$ in $Q_N$. Since $ij \in E(H)$, by definition of $F(G,\PP,H)$ we have that $u$ and $v$ are both adjacent to $w$ in $G$ and hence $dist_G(u,v) \le 2$. 

(ii) Assume that $u \in V_i$ and $v \in V_j$. Since in $Q_N$ the vertex $v$ has at most $6$ neighbors in $V_i$, we know that there is a non-neighbor $w$ of $v$ in $V_i$ (again, with respect to $Q_N$). Since $ij \in E(H)$, we know that $vw \in E(G)$. Since $w$ and $u$ are both in $V_u$, by (i) we have $dist_G(u,w) \le 2$, and so $dist_G(u,v) \le 3$, as desired.
\end{proof}

\begin{lemma}
\label{lem:conn}
    Let $(G,\PP,H)$ a $k$-flip structure be such that $F(G,\PP,H) = Q_N$ and $|V_i| > 12$ for each $i \in [k]$ with $V_i\not=\emptyset$. Then $G$ is connected.
\end{lemma}
\begin{proof}
    Assume for contradiction that $G$ is disconnected. Since $Q_N$ is connected, the transition from $G$ to $Q_N$ has to make vertices from different components of $G$ adjacent in $Q_N$. Therefore, there exist connected components $C$ and $C'$ of $G$ such that there are vertices $u \in C$ and $v \in C'$ with $u \in V_i$, $v \in V_j$ and $ij \in E(H)$. Since both $V_i$ and $V_j$ are non-empty, we have  $|V_i| > 12$ and $|V_j| > 12$ by the assumptions of the lemma, and so by Lemma~\ref{lem:small_dist} the distance between $u$ and $v$ in $G$ is at most $3$. Thus, in particular both $u$ and $v$ are in the same connected component, a contradiction.
\end{proof}

By combining the last three lemmas we obtain the following lemma that summarizes properties of a class $\C$ of $k$-flip structures for which we have $\Q \subseteq F(\C)$. Recall that $I(H)$ denotes the set of vertices of $H$ that do not have a neighbor in $H$ and do not have loops. 

\begin{lemma}
\label{lem:properties}
    Let $\E$ be a hereditary class of $k$-flip structures such that $\Q \subseteq F(\E)$. Then $\E$ contains, for each $N \in \N$, a $k$-flip structure $(G,\PP,H)$ such that $F(G,\PP,H) = Q_N$ and
    \begin{enumerate}[(i)]
        \item for each non-empty part $V_i \in \PP$ we have $|V_i| > 12$,
        \item the graph $G$ is connected, and
        \item if $C$ is a connected component of $H - I(H)$, then $G[\cup_{i \in C}V_i]$ has diameter at most $3|C|$.
    \end{enumerate}
\end{lemma}
\begin{proof}
(i) Let $N \in \N$ be arbitrary.
Let $(G,\PP,H)$ be such that  $F(G,\PP,H) = Q_{3^kN}$. By repeated application of Lemma~\ref{lem:only_large_classes} we can extract from $(G,\PP,H)$ a $k$-flip substructure $(G', \PP', H')$ such that $F(G', \PP', H') = Q_N$ in which each $V_i' \in \PP'$ is either empty or has more than $12$ vertices.

 (ii) Since all non-empty parts of $\PP'$ have more than $12$ vertices, by Lemma~\ref{lem:conn} we know that $G'$ is connected.

 (iii) Let $C$ is a connected component of $H'- I(H')$. Let $u,v \in \cup_{i \in C}V_i$. Since $C$ is connected, there is a path $P$ in $C$ between $\rho(u)$ to $\rho(v)$; let $m$ be the length of $P$. Path $P$ corresponds to a sequence $V_{i_1},\ldots, V_{i_m}$ of non-empty sets of sets from $\PP'$, where $i_j$ is the $j$-th vertex of $P$, and $u \in V_{i_1} = \rho(u)$, $v \in V_{i_m} = \rho(v)$. From each $V_{i_j}$ pick we pick one vertex $w_j$, where we set $w_1 = u$, $w_m = v$, and for $j \in [2, m-1]$ we pick $w_j \in V_{i_k}$ arbitrarily.
Since every $V_{i_j}$ has more than $12$ vertices, we can apply Lemma~\ref{lem:small_dist} to every pair $w_j, w_{j+1}$, and we obtain a walk of length at most $3m$ from $u$ to $v$. From this walk we can extract a path of length at most $3m$ from $u$ to $v$. Since $m \le |C|$, this finishes the proof.
\end{proof}

\subsection{Finishing the proof}

The proof of Theorem~\ref{thm:reduction_bd_range} will follow easily from the following lemma.
\begin{lemma}
\label{lem:flip_bd_range}
Let $\E$ be a hereditary class of $k$-flip structures such that $\Q \subseteq F(\E)$. Then there exists a transduction $\ST$ of bounded range such that $\Q \subseteq \ST(\D)$, where $\D = \{G~|~(G, \PP, H ) \in \E\}$.
\end{lemma}
\begin{proof}
For every $k$-flip structure $(G,\PP,H)$ we define a number $\lambda(G, \PP, H)$ by setting

$$\lambda(G,\PP,H):=  \max\{r \in \N_0~|~\exists v \in V(G) \text{ with } N^G_{r}[v] 
\subseteq  \cup_{i \in I(H)} V_i\}$$

We distinguish two cases:

1) For every $n$ there is a $k$-flip structure $(G,\PP,H) \in \E$  with $\lambda(G,\PP,H) \ge n$ such that $F(G,\PP,H) = Q_N$ for some $N$. In this case, by Lemma~\ref{lem:induced_subgrid} we know that $G$ contains $Q_{n/3}$ as an induced subgraph. Consequently, graphs from the class $\D$ from the statement of the lemma contain all cubes as induced subgraphs. 
Thus, the transduction $\ST$ that does not do anything except for taking induced subgraphs (meaning that there is no coloring step and we use the trivial interpretation that leaves the graph unchanged) has the property that $\Q \subseteq \ST(\D)$. Clearly, the interpretation formula $\psi(x,y) = E(x,y)$ is of bounded range, as desired.

2) There exists a number $\alpha$ such that $\lambda(G,\PP,H) \le \alpha$ for all $(G,\PP,H) \in \E$ with $F(G,\PP,H) = Q_N$. This means that for every $(G,\PP,H) \in \E$ and every $v \in V(G)$ there exists $i \in V(H) \setminus I(H)$ and $w \in V_i$ such that $dist_G(v,w) \le \alpha$. Our goal will be to show that in this case, for every $(G,\PP,H) \in \E$ with $F(G, \PP, H) = Q_N$ the graph $G$ has  diameter at most $2k\alpha + 4k$. Then, we can consider as $\ST$ the following transduction that produces $Q_N$ from $G$: We mark vertices of $G$ with unary predicates to encode the partition $\PP$ and use one extra predicate to remember the graph $H$ from $(G, \PP, H)$ (this works because there are only finitely many possible graphs on $k$ vertices; we can for example mark every vertex of $G$ with the predicate that encodes $H$). Then the interpretation formula $\psi(x,y)$ checks whether the diameter of $G$ is at most $2k\alpha + 4k$, and if yes, it produces $F(G,\PP,H)$ from $G$. Otherwise the formula just keeps the original edges of $G$. This formula is clearly of range at most $2k\alpha + 4k$.

For a connected component $C$ of $H-I(H)$, we will denote by $U(C)$ the subset of $V(G)$ given by $\cup_{i\in C} V_i$. 
We define an auxiliary graph $K$ as follows: The vertices of $K$ are the connected components of $H - I(H)$ and there is an edge between two components $C$, $C'$ if there is a path of length at most $2\alpha+1$ in $G$ between a vertex in $U(C)$ and a vertex in $U(C')$. We will show that  graph $K$ is connected. This will imply that $G$ has diameter at most $2k\alpha + 4k$ as follows. Let $u,v \in V(G)$ be arbitrary. First consider the case when $u \in U(C)$ and $v \in U(C')$ for some connected components of $H - I(H)$. Since $K$ is connected, there exists path from $C$ to $C'$ in $K$ of length at most $k$ (since $k$ is the number of vertices in $H$). Using the fact that between any two consecutive components on the path there is a path of length at most $2\alpha+1$ in $G$, and that by Lemma~\ref{lem:properties} there is a path of length at most $3$ between any two vertices within any connected component, we can connect $u$ and $v$ by a path of length at most $(k-1)(2\alpha + 1) + 3k$. Second, if $u \in \cup_{i \in I(H)}C_i$ or $v \in \cup_{i \in I(H)}C_i$ (or both), we know that there exists $i \in V(H) \setminus I(H)$ and $u'\in U(C_i)$ such that $dist_G(u,u') \le \alpha$, and similarly we have $dist_G(v,v') \le \alpha$ for some $j$ and $v' \in V_j$. As before, we know that $dist_G(u',v')\le (k-1)(2\alpha +1) + 3k$, and so we can upper bound $dist_G(u,v)$ by $2k\alpha + 4k$.

We now proceed with proving that the graph $K$ is connected.
If $K$ has only one vertex (that is, if $H-I(H)$ has only one connected component), then $K$ is clearly connected. From now on we will assume that $K$ has more than $1$ vertex.
Let $\F = \{C_1,\ldots, C_m\}$ be the connected components of $H - I(H)$. Let $S_1,S_2$ be any bipartition of $\F$ with both $S_1$, $S_2$ non-empty. Let $$W_1 = \bigcup_{C \in S_1}U(C) \text{ and } W_2 = \bigcup_{C \in S_2}U(C).$$ We define the distance $dist(S_1,S_2)$ between $S_1$ and $S_2$ to be the minimum of $dist_G(u,v)$ over all $u \in W_1$ and $v \in W_2$. Note that since $G$ is connected, we know that $dist(S_1,S_2) < \infty$. Let $u \in W_1$ and $v \in W_2$ be such that they minimize $dist(u,v)$ among all such $u,v$. We claim that $dist(S_1,S_2) \le 2\alpha + 1$. Let $P$ be a path of shortest length from $u$ to $v$ in $G$. 
 Let $w$ be a  vertex of $P$ such that $|dist_G(u,w) - dist_G(w,v)| \le 1$; such vertex exists in any path. From our assumption $\lambda(G,\PP,H) \le \alpha$ we know that there exists $w' \in W_1\cup W_2$ such that $dist_G(w,w') \le \alpha$. If $w' \in W_1$, then $dist_G(u,w) \le dist_G(w',w)$, because if $dist_G(w',w) < dist_G(u,w)$ then by concatenating paths from $w'$ to $w$ and from $w$ to $v$ we would get a path from $W_1$ to $W_2$ that is shorter than $P$). By analogous reasoning, if $w' \in W_2$, then $dist_G(v,w) \le dist_G(w',w)$. Thus, we get that 
 one of $dist_G(u,w)$ or $dist_G(w,v)$ is at most $\alpha$, and since the difference between these two distances is at most $1$, we have $dist_G(u,v) \le 2\alpha + 1$.
\end{proof}

\begin{proof}[Proof of Theorem~\ref{thm:reduction_bd_range}]
    Let $\C$ be a class of graphs such that $Q \subseteq \T(\C)$ for some transduction $T$. By Theorem~\ref{thm:gaifman_bd_range}, there exists a transduction $R$ of bounded range and $k \in \N$ such that every $Q_N$ is a $k$-flip of some graph $G$ from $R(\C)$. 
    Let $\E$ be the class of all possible $k$-flip structures $(G,\PP,H)$ such that $G \in R(\C)$. Then clearly $\Q \subseteq F(\E)$ and $\E$ is hereditary, since $R(\C)$ is hereditary.  
    By Lemma~\ref{lem:flip_bd_range} there exists a transduction $S$ of bounded range such that $\Q \subseteq S(R(\C))$. Since the composition of two transductions of bounded range is a transduction of bounded range, we get that $\Q \subseteq T'(\C)$, where $T' = S \circ R$ is a transduction of bounded range, as desired.
\end{proof}

\section{Proof Theorem~\ref{thm:main}}
\label{sec:finish}

Following~\cite{Dvorak_grids}, we define $\wh{Q}_N$ to be the supergraph of $Q_N$ obtained by adding all non-decreasing diagonals to all $1\times 1 \times 1$ cubes in $Q_N$. That is, besides keeping all edges of $Q_N$, in $\wh{Q}_N$ every vertex $(i,j,k)$ is adjacent also to vertices $(i+1,j+1,k)$, $(i+1,j, k+1)$, $(i, j+1, k+1)$ and $(i+1,j+1, k+1)$, whenever these vertices exist. We denote the class consisting of all $Q_N$ for $N \in \N$ by $\wh{\Q}$. 

We will use on the following result of Berger, Dvo\v{r}\'ak and Norin:
\begin{theorem}[\cite{Dvorak_grids}]
\label{thm:grids_dvorak}
    For every $t$ there exists $N$ such that for every bipartition $A_1, A_2$ of $V(\wh{Q}_N)$ we have that $tw(\wh{Q}_N[A_1]) > t$ or $tw(\wh{Q}_N[A_2]) > t$. 
\end{theorem}

\begin{lemma}
\label{lem:diag_cubes}
    There exists a transduction $\RT$ of range $3$ such that $\wh{\Q}_N \in R(\Q_N)$ for every $N$.
\end{lemma}
\begin{proof}
    Recall that on $Q_N$ we have three projection functions $\pi_1,\pi_2,\pi_3$ such that for $v = (i,j,k)$ we have $\pi_1(v) = 1$, $\pi_2(v) = j$ and $\pi_3(v) = k$. We say that vertex $v$ is a $x$-successor of vertex $u$ if $\pi_1(v) = \pi_1(u) + 1$ and $\pi_2(v) = \pi_2(u)$, $\pi_3(v) = \pi_3(u)$. We define $y$-successor and $z$-successor analogously. Then the defining conditions on diagonal edges of $\wh{Q}_N$ can be expressed in terms of $x$-, $y$-, and $z$-successors. For example, for $u = (i,j,k)$ and $v = (i+1,j+1,k)$ we have that there exists $w$ such that $w$ is a $x$-successor of $u$ and $v$ is a $y$-successor of $w$. 

    We will show that one can equip $Q_N$ in such a way that being a $x$-, $y$-, and $z$-successor can be expressed by a first-order formula.
   We assign unary predicates $X_0,X_1,X_2, Y_0,Y_1,Y_2,\\ Z_0,Z_1,Z_2$ to vertices of $Q_N$ as follows. For $a,b,c \in \{0,1,2\}$ we mark vertex  $(i,j,k)$ with $X_a$ if $i \equiv a ~(mod ~3)$, with $Y_b$ if $j \equiv a ~(mod ~3)$, and with $Z_c$ if $k \equiv a ~(mod ~3)$. Then we can check  for two vertices $u,v$ of $Q_N$ whether $v$ is a $x$-successor of $u$ by setting
   $$succ_X(x,y):= E(x,y) \land \bigvee_{\substack{a,b \in \{0,1,2\} \\ a \equiv b + 1 (mod 3)}}   (X_a(x) \land X_b(y) )$$
   Formulas $succ_y(x,y)$ and $succ_z(x,y)$ are defined analogously. Then, as indicated above, one can express for example whether $\pi_1(u) = \pi_1(v)+1$, $\pi_2(u) = \pi_2(v) + 1$ and $\pi_3(u) = \pi_3(v)$ by $\exists z. succ_x(x,z) \land succ_y(z,y)$. One can therefore write a FO formula $\beta_{diag}(x,y)$ that expresses whether two vertices $u,v$ of $Q_N$ should form a non-decreasing diagonal, and then set $\psi(x,y):=E(x,y) \lor \beta_{diag}(x,y)$. Then we have $\wh{Q}_N = \psi(Q_N)$, and formula $\psi(x,y)$ has range $3$, as desired.
\end{proof}


\begin{proof}[Proof of Theorem~\ref{thm:main}]
  Let $\C$ be a class of graph of bounded genus and assume for contradiction that $\Q \subseteq \T(\C)$ for some transduction $\T$. Then, by Theorem~\ref{thm:reduction_bd_range} there exists a transduction $\T'$ of bounded range such that $\Q \subseteq \T'(\C)$. By Lemma~\ref{lem:diag_cubes} we know that $\wh{\Q} \subseteq \RT(\T'(\C))$, and since the composition of two bounded range transduction is a transduction of bounded range, we have that $\wh{\Q} \subseteq S(\C))$, where $\ST = \RT \circ \T$ is a transduction of bounded range.  Since $\C$ has bounded genus, by Lemma~\ref{lem:f-decompositions} we know that $\wh{\Q}$ has $f$-decompositions for some function $f$. Since the maximum degree of graphs from $\wh{\Q}$ is at most $14$, we have that $K_{15,15}$ is not a subgraph of any graph from $\wh{\Q}$. 
  Set $t:=g(f(1),t)$, where $g$ is the function from Theorem~\ref{thm:cw_sparse_tw}. Let $N$ be obtained from applying Theorem~\ref{thm:grids_dvorak} to $t$. Let $V_1, \ldots, V_\ell$ be a $f$-decomposition of $\wh{Q}_N$. 
  Set 
  $$A_1:= \bigcup_{\substack{i \in [\ell] \\ \text{$i$ even}}} V_i $$
  and $A_2:=V(\wh{Q}_N)\setminus A_1$. 
  Since for each $i$ the graph $\wh{Q}_N[V_i]$ has clique-width at most $f(1)$ and is $K_{15,15}$-free, its  treewidth is at most $t$. Since  $\wh{Q}_N[A_1]$ is a disjoint union of graphs $\wh{Q}_N[V_i]$, we have that $tw(\wh{Q}_N[A_1]) \le t$. The same argument yields that $tw(\wh{Q}_N[A_2]) \le t$. This is a contradiction to Theorem~\ref{thm:grids_dvorak}.
\end{proof}

\bibliography{biblio}
\bibliographystyle{abbrv}

\end{document}